\documentclass[conference,twocolumn]{IEEEtran}\IEEEoverridecommandlockouts

\usepackage{amssymb}
\usepackage{amsmath}
\usepackage{amsfonts}
\usepackage{graphicx}
\usepackage{algorithm}
\usepackage{algorithmic}
\usepackage{cite}
\usepackage{epstopdf}
\usepackage{color}
\usepackage{multirow}
\usepackage{enumerate}

\newtheorem{lemma}{\textbf{Lemma}}
\newtheorem{proposition}{\textbf{Proposition}}

\newtheorem{remark}{\textbf{Remark}}
 \begin{document}
\title{Optimal Caching Placement for D2D Assisted Wireless Caching Networks }
\author{Jun Rao$^{*}$, Hao Feng$^{*}$, Chenchen Yang$^{*}$, Zhiyong Chen$^{*\dag}$, and Bin Xia$^{*}$\\
$^{*}$ Department of Electronic Engineering, Shanghai Jiao Tong University, Shanghai, P. R. China\\
$^{\dag}$ Cooperative Medianet Innovation Center, Shanghai, China\\
Email: {\{sjtu\_jun, fenghao, zhanchifeixiang, zhiyongchen, bxia\}@sjtu.edu.cn}\\
}
\maketitle

\begin{abstract}
In this paper, we devise the optimal caching placement to maximize the offloading probability for a two-tier wireless caching system, where the helpers and a part of users have caching ability. The offloading comes from the local caching, D2D sharing and the helper transmission. In particular, to maximize the offloading probability we reformulate the caching placement problem for users and helpers into a difference of convex (DC) problem which can be effectively solved by DC programming. Moreover, we analyze the two extreme cases where there is only help-tier caching network and only user-tier. Specifically, the placement problem for the helper-tier caching network is reduced to a convex problem, and can be effectively solved by the classical water-filling method. We notice that users and helpers prefer to cache popular contents under low node density and prefer to cache different contents evenly under high node density.  Simulation results indicate a great performance gain of the proposed caching placement over existing approaches.
\end{abstract}
\section{Introduction}
 As more and more different types of smart devices are produced and applied in people's daily life, wireless traffic demand has experienced an unprecedented growth. Cisco's most recent report\cite{cisco} forecasts that the  mobile multimedia data will grow at a compound annual growth rate of more than $60\%$. On the other hand, users demand for multimedia contents is highly redundant, i.e., a few popular contents account for a majority of all requests\cite{Zipf}. Therefore, caching popular contents at various nodes in the network is a promising approach to alleviate the network bottleneck\cite{IWCT}.

For the wireless caching systems where helpers  (WiFi, femtocells) have high storage capacity, the performance  depends heavily on the adopted caching replacement. The caching placement for helpers is firstly investigated in \cite{4femetocaching} to minimize the downloading time, where both uncoded and coded cases are considered.  It is shown
that the optimization problem for the uncoded case is NP-hard. In addition,  \cite{cachingplacement} considers  the channel fading factor and develops the caching placement to minimize the average bit error rate, where the optimal caching placement is to balance between the channel diversity gain and the  caching diversity gain.  Moreover,  the problem of optimal
MDS-encoded caching placement at the wireless edge is investigated in \cite{MDS} to minimize
the backhaul rate in heterogeneous networks. However, all above analyses \cite{4femetocaching,cachingplacement,MDS} are  based on the fixed topology between users and helpers. In  \cite{geographic}, more realistic network models are adopted to characterize the stochastic natures of geographic location and the corresponding optimal caching placements are derived according to the total hit probability.

On the other hand, the potential cache capacity at user side can also be exploited, e.g., local cache offloading or D2D sharing. Various works have been done on the caching placement at user side. In \cite{outage}, the D2D outage-throughput tradeoff problem is investigated and the optimal scaling laws are characterized.  \cite{USC} analyzes the scaling
behavior of the throughput with the number of devices per cell under Zipf distributed content request probability with exponent $\gamma_r$, and concludes that  the optimal cache distribution is also a Zipf distribution with a different  value $\gamma_c$.  By modeling the mobile devices as a homogeneous Poisson Point Processes (PPP), \cite{malakGlobecom} derives the optimal cache distribution resulting in the total probability of content delivery is maximized. However, the local offloading probability is not considered in their analysis. In addition, coded caching is also an effective approach to exploit the content diversity \cite{codedcache1}. By caching contents partially at user side according to the developed caching distribution during the first phase, a coded multicasting opportunity can be created even for different content requests in the second phase. Moreover, \cite{Hierarchical} further proposes the hierarchical coded caching
to address the joint caching placement problem at both users and helpers. However, these  analyses \cite{codedcache1,codedcache2,Hierarchical} are  based on the fixed topology which is not suitable for the user mobility scenario.

Despite the aforementioned studies, to the best of our knowledge, the optimal caching placements for both helpers and users under realistic network models  remain unsolved to date. Thus in this paper, we consider a two-tier caching system, where the helpers and users are spatially distributed according to two mutually independent homogeneous Poisson Point Processes (PPPs) with different densities \cite{YangCC}. In order to alleviate the traffic load in the cellular network, we aim to develop an optimal caching placement scheme to maximize the offloading probability, where the offloading includes self-offloading, D2D-offloading and helper-offloading. More details along with the main contributions are as follows:

\begin{itemize}
  \item
 We consider a D2D assisted two-tier wireless caching network consisting of users and helpers where the offloading comes from self-offloading, D2D-offloading and helper-offloading. Different from \cite{malakGlobecom}, we take self-offloading events into consideration. Moreover, the practical assumption that only a part of users has caching ability is considered.
  \item
  We formulate the total offloading probability of caching placement in the two-tier wireless networks and adopt the DC programming to solve the non-convex maximization problem. In addition, we notice that users and helpers ought to cache the popular contents while the density is low and ought to cache different contents while the density is high. And our proposed caching placement can achieve a balance between them.
  \item
  The two extreme cases for one-tier caching systems are considered. In absence of user caching ability, we formulate the  caching placement for helper-tier as a convex problem, and can be effectively solved by the classical waterfilling method; In absence of helper caching ability, the  caching placement for users is also formulated into a convex problem. Furthermore, we combine the solutions of the two cases as a non-joint optimal caching placement and compare it with the proposed placement.
\end{itemize}
\section{System model and content access protocol}
In this section, we first introduce  the  two-tiered caching  system as illustrated in Fig. \ref{fig:system}, where the helpers and users are spatially distributed according to two mutually independent homogeneous Poisson  Point Processes (PPPs) with density $\lambda_{\text{H}}$ and $\lambda_{\text{UE}}$, respectively. Then  the content access protocol is provided.
\subsection{System Model}
\subsubsection{Content module}
The content library consists of $N$ contents. The popularity distribution vector of the contents is denoted by $\mathbf{q}=\{q_1,\ldots,q_N\}$, where $q_i$ is the access probability for the $i$-th content. In this paper, we characterize the popularity distribution as a Zipf distribution with parameter $\gamma$\cite{Zipf}.  If we arrange contents in descending order of popularity, the popularity of the $i$-th ranked content is \cite{8Push-Based}
\begin{equation}
q_i={\frac{1/i^{\gamma}}{\sum_{j=1}^{N}1/j^{\gamma}}},
\end{equation}
where $\gamma$ governs the skewness of the popularity. The popularity is uniform over contents for $\gamma=0$, and becomes more skewed as $\gamma$ grows. For simplicity, we assume all the $N$ contents are of equal size $L$.
\subsubsection{Network module}
 In addition to the macro base stations (BSs), the network module also consists of the helpers with caching ability, where  helpers could successfully send  the contents in its local cache to  requesting users within radius $R_{\text{H}}$ at relatively low cost.   For simplicity,  we assume the caching capacity for all helpers are the same, denoted by  $M_{\text{H}}L$, where $M_{\text{H}}<N$. Therefore, the helper can cache up to $M_{\text{H}}$ different contents entirely. Also we assume a content can  only be cached entirely rather than partially. Denote the caching placement at the helpers for each content as $\mathbf{P_{\text{H}}}=[p_{1}^{\text{H}},p_{2}^{\text{H}},\ldots,p_{N}^{\text{H}}]$, where $p_i^{\text{H}}$  is  the proportion of helpers caching the $i$-th content and  $0 \leq p_i^{\text{H}}\leq 1$ for $i=1,2,\ldots,N$.  The cache storage constraint at the helpers can then be  written as  $\sum\limits_{i=1}^{N}p_i^{\text{H}}\leq M_{\text{H}}$. The helpers caching the $i$-th content also follow a PPP with density $\lambda_{\text{H}}p_{i}^{\text{H}}$.
\subsubsection{User module}
We assume part of the users having caching ability. Let $\alpha$ denote the proportion of cache-enabled users, where $0\leq \alpha \leq 1$.  The cache-enabled users  also follow a thinning homogeneous  PPP with density $\alpha\lambda_{\text{UE}}$. For simplicity,  we  assume the caching capacity for the cache-enabled users are the same, denoted by  $M_{\text{UE}}L$. Therefore, cache-enabled users can cache up to $M_{\text{UE}}$ different contents entirely in its local cache. Moreover, a device to device (D2D) communication can be established if the distance between the requesting user and the user caching the desired content is less than $R_{\text{UE}}$, where $R_{\text{UE}}<R_{\text{H}}$ due to the transmitting power. Let $\mathbf{P_{\text{UE}}}=[p_{1}^{\text{UE}},p_{2}^{\text{UE}},\ldots,p_{N}^{\text{UE}}]$ denote the caching placement at the cache-enabled users for each content, where $p_i^{\text{UE}}$  is  the proportion of users caching the $i$-th content, and  $0 \leq p_i^{\text{UE}}\leq 1$ for $i=1,2,\ldots,N$.  The cache storage constraint at the cache-enabled users  can then be  written as  $\sum\limits_{i=1}^{N}p_i^{\text{UE}}\leq M_{\text{UE}}$. Therefore, the users caching the $i$-th content also follows a PPP with density $\alpha\lambda_{\text{UE}}p_{i}^{\text{UE}}$.

\begin{figure}[t]
\centering
\includegraphics[height=2.2in]{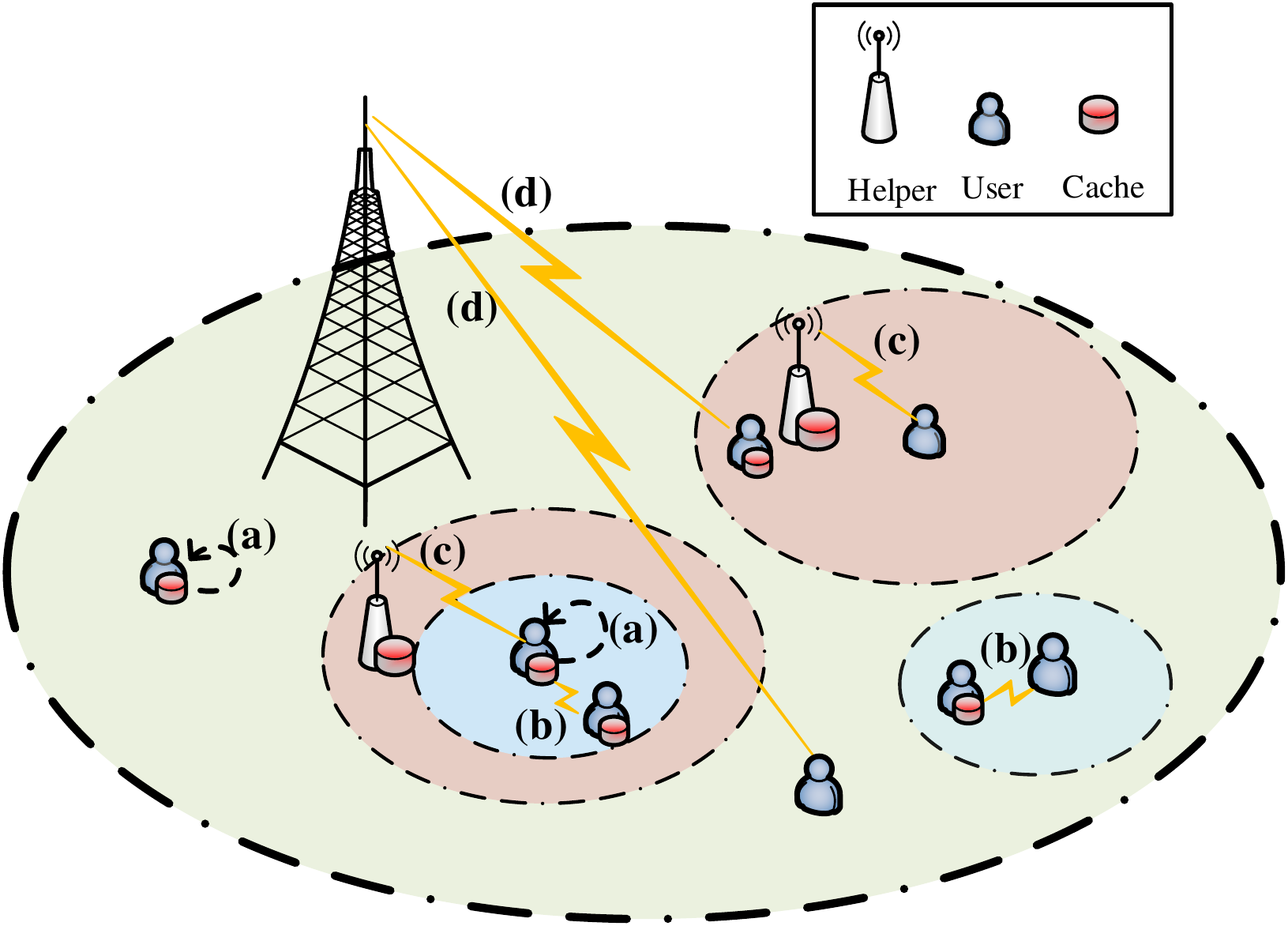}
\vspace{-0.3cm}
\caption{System model of the D2D assisted wireless caching system, where (a), (b), (c) and (d) stand for Self-offload, D2D-offload, Helper-offload and Celluar-response, respectively.}
\vspace{-0.2cm}
\label{fig:system}
\end{figure}

\subsection{Content Access Protocol}
As indicated in Fig. \ref{fig:system}, the content access protocol is as follows:
 \begin{enumerate}[(a)]
  \item
\textbf{Self-offloading}:
When a content request occurs, the user first checks its local storage whether the desired content has been stored in it. The request will be satisfied and offloaded immediately if the user has cached the desired content in its local storage space. We term it as ``Self-offloading".
\item
\textbf{D2D-offloading}:
If the exact content has not been cached in the local storage or the requesting user does not have cache ability, the user will turn to search near devices for the desired content. If there is at least one users have stored the desired content within the radius $R_{\text{UE}}$. The request would be met and offloaded by establishing a D2D communication, termed as ``D2D-offloading".
\item
\textbf{Helper-offloading}:
In addition to ``D2D-offloading", if there is at least one helper have stored the desired content within $R_{\text{H}}$,  the request would be satisfied and offloaded by the helper transmission, termed as ``Helper-offloading".
\item
\textbf{Cellular-response}:
If the request could not be offloaded via local cache, D2D or the helpers\, then it need to be forwarded to the cellular base station and the cellular network transmits the requested content in response.
\end{enumerate}
\section{offloading probability and Problem formulation}
In this paper, in order to alleviate the traffic load from the cellular network, our goal is to find the optimal caching placement to maximize the offloading probability. Therefore, we first analyze the offloading probability for the D2D assisted wireless caching network. Then, the optimal caching placement problem is formulated.
\subsection{Offloading probability analysis}
 For a PPP distribution with density $\lambda$, the probability that there are $n$ devices in the area within a radius $r$ is:

 \begin{equation}
 \mathbb{F}(n,r,\lambda)= \frac{(\pi r^2\lambda)^n}{n!} e^{-\pi r^2 \lambda}
 \end{equation}

Therefore,  for a reference user located at the origin, the probability of at least another user caching the $i$-th content within the transmission range is
\begin{align}
P_{i,\text{off}}^{\text{D2D}}=1-\mathbb{F}(0,R_{\text{UE}},\alpha \lambda_{\text{UE}} p_{i}^{\text{UE}})=1-e^{-\pi\alpha \lambda_{\text{UE}} p_{i}^{\text{UE}} {R_{\text{UE}}}^2}.
\end{align}
Similarly, the probability of at least one helper caching the $i$-th content within the radius $R_{\text{H}}$ is
\begin{equation}
P_{i,\text{off}}^{\text{H}}=1-\mathbb{F}(0,R_{\text{H}},\lambda_{\text{H}} p_{i}^{\text{H}})=1-e^{-\pi\lambda_{\text{H}} p_{i}^{\text{H}} {R_{\text{H}}}^2}.
\end{equation}
The offloading probability for the $i-$th content of cache-unabled users, i.e the probability at least one helper or one user caching the $i$-th content is:
\begin{equation}
\begin{split}
P_{i,\text{NC}}=&1-(1-P_{i,\text{off}}^{\text{D2D}})(1-P_{i,\text{off}}^{\text{H}})\\
=&1-e^{-(\pi\alpha \lambda_{\text{UE}} p_{i}^{\text{UE}} {R_{\text{UE}}}^2+\pi\lambda_{\text{H}} p_{i}^{\text{H}} {R_{\text{H}}}^2)}.
\end{split}
\end{equation}
The corresponding offloading probability of the cache-enabled users for the $i$-th content is
\begin{equation}
P_{i,\text{C}}=p_{i}^{\text{UE}}+(1-p_{i}^{\text{UE}})P_{i,\text{NC}}.
\end{equation}
Therefore, the offloading probability for the $i$-th content becomes
\begin{equation}
\begin{split}
 P_{i,\text{off}}=&\alpha{P_{i,\text{C}}}+(1-\alpha )P_{i,\text{NC}}\\
=&1-(1-\alpha{p_{i}^{\text{UE}}})e^{-(\pi\alpha \lambda_{\text{UE}} p_{i}^{\text{UE}} {R_{\text{UE}}}^2+\pi\lambda_{\text{H}} p_{i}^{\text{H}} {R_{\text{H}}}^2)}.
\end{split}
\end{equation}

The total offloading probability for the D2D assisted wireless caching system becomes
\begin{equation}
 P_{\text{off}}=\sum_{i=1}^{N} q_{i}P_{i,\text{off}},
\end{equation}
while more data offloaded by the wireless caching network, less data needs to be sent via the cellular network, alleviating the traffic load for the cellular network.
\subsection{Problem Formulation}
Let $\mathbf{P}=[\mathbf{P_{\text{H}}}\quad \mathbf{P_{\text{UE}}}]$ denotes the caching placement at helper and user sides.
The optimal caching placement that maximizes the offloading probability for the wireless caching network can be formulated as
\begin{align}
    & \max_{\mathbf{P}} ~~\sum_{i=1}^{N}q_{i}P_{i,\text{off}}\\
    & \textup{s.t.}~~\begin{cases}\sum\limits_{i=1}^{N} {p_{i}^{\text{UE}}}\leq M_{\text{UE}}\\
  \sum\limits_{i=1}^{N} {p_{i}^{\text{H}}}\leq M_{\text{H}}\\
   0 \leq p_i^{\text{UE}}\leq 1,  i\in\{1,\ldots,N\}\\
  0 \leq p_i^{\text{H}}\leq 1, i\in\{1,\ldots,N\}\\
    \end{cases}.
\end{align}
\section{DC Programming for Caching Placement Optimization}\label{sec:dc}

In this section, we adopt the difference of convex (DC) program to solve the above problem. The maximization problem is equivalent to the following minimization problem:
\begin{align}\label{pro:equvalent_min}
    & \min_{\mathbf{P}} ~~-\sum_{i=1}^{N}q_{i}P_{i,\text{off}}\\
    & \textup{s.t.}~~\begin{cases}\label{pro:st}\sum\limits_{i=1}^{N} {p_{i}^{\text{UE}}}\leq M_{\text{UE}}\\
  \sum\limits_{i=1}^{N} {p_{i}^{\text{H}}}\leq M_{\text{H}}\\
   0 \leq p_i^{\text{UE}}\leq 1,  i\in\{1,\ldots,N\}\\
  0 \leq p_i^{\text{H}}\leq 1, i\in\{1,\ldots,N\}\\
    \end{cases}.
\end{align}

Let $F(\mathbf{P})=-\sum_{i=1}^{N}q_{i}P_{i,\text{off}}$ denote the objective function in problem (\ref{pro:equvalent_min}) and it can be easily verified that the hessian matrix of $F(\mathbf{P})$ is not positive definite and hence $F(\mathbf{P})$ is non-convex.

Let $H(\mathbf{P})=\sum_{i=1}^{n}q_ih_i$, where  $h_i=\alpha\pi\lambda_{\text{H}}{R_{\text{H}}}^2({p_i^{\text{UE}}}^2+{p_i^{\text{H}}}^2)$. Denote $G(\mathbf{P})=F(\mathbf{P})+H(\mathbf{P})$, we then have the following proposition.
 \begin{proposition}
    $H(\mathbf{P})$ and $G(\mathbf{P})$ are both convex of $\mathbf{P}$.
 \end{proposition}

\begin{proof}
Let $A_i$ denote the hessian matrix of $ h_{i}$

$A_i=\begin{bmatrix}
    \frac{\partial^{2}{h_{i}}}{\partial {(p_i^{\text{UE}})}^{2}} & \frac{\partial^{2}{h_{i}}}{\partial {p_i^{\text{UE}}}\partial {p_i^{\text{H}}}}\\
\frac{\partial^{2}{h_{i}}}{\partial {p_i^{\text{H}}}\partial {p_i^{\text{UE}}}}&\frac{\partial^{2}{h_{i}}}{\partial {(p_i^{\text{H}})}^{2}}
  \end{bmatrix}=\begin{bmatrix}
    2\alpha\pi\lambda_{\text{H}}{R_{\text{H}}}^2& 0 \\
    0 & 2\alpha\pi\lambda_{\text{H}}{R_{\text{H}}}^2
  \end{bmatrix}$.

Hence the matrix $A_i$ is positive definite and $h_i$ is convex. Since the linear combination of convex functions is also convex,  $H(\mathbf{P})$ is  convex. Similarly, we have the hessian matrix of $G(\mathbf{P})$ is definite and $G(\mathbf{P})$  convex of $\mathbf{P}$.
\end{proof}

 Hence, $ F(\mathbf{P})$ can be written as a difference of the following two convex functions:
 \begin{equation}
 F(\mathbf{P})=G(\mathbf{P})-H(\mathbf{P}).
 \end{equation}

Therefore, we adopt the DC programming to solve this problem. DC programming is a quick convergence programming which can obtain a partial optimal solution and sometimes the global optimal solution of a non-convex function\cite{dc}. Since  $\frac{\partial{H(\mathbf{P})}}{\partial {\mathbf{P}}}$ is continuous and the constraint of problem (\ref{pro:equvalent_min}) is a convex set, the DC programming can be simply described in Algorithm $1$. The result will be illustrated in section \ref{sec:result}.

\begin{algorithm}\label{dc}
    \caption{DC programming for  caching placement}
    \begin{algorithmic}[1]
           \STATE   initial value: $\mathbf{P}_{0}^{\text{UE}}=\frac{M_\text{UE}}{\text{N}},
                         \mathbf{P}_{0}^{\text{H}}=\frac{M_\text{H}}{\text{N}}$;\\
           \STATE  solve the convex optimization problem:

               $\min G(\mathbf{P})-H(\mathbf{P}_{k})-(\mathbf{P}
                -\mathbf{P}_{k}) \frac{\partial{H(\mathbf{P}_{k})}}{\partial {\mathbf{P}}}~~\textup{s.t.}(\ref{pro:st})$;\\
        \STATE     the solution of step $2$ is $\mathbf P_{{k+1}}$;\\
      \STATE    if $ ||F(\mathbf{P}_{k})-F(\mathbf{P}_{{k+1}})|| \leq \varepsilon$ or $||\mathbf P_{k}-\mathbf P_{{k+1}}||\leq \varepsilon$,$\mathbf P_{k}$ is the optimal solution of $F(\mathbf{P})$; otherwise,return to step$2$ ;\\
      \STATE    RETURN:the result is:$F(\mathbf{P}_{k})$,the solution is: $\mathbf{P}_{k}$;
    \end{algorithmic}
\end{algorithm}

\section{extreme case analyses}
In this section, we consider the  caching problem under extreme cases where only one tier of the caching system is considered and the optimal caching placement can be calculated. We analyze the caching placement of the two extreme cases and combine the solutions as a baseline.
\subsection{${\alpha}=0$: \textbf{helper-tier caching network}}
\subsubsection{problem formulation}
In this case, all  users have no caching ability and we only need to optimize the caching placement $\mathbf{P_{H}}$ at helper side. The  offloading probability for the $i$-th content is reduced to
\begin{equation}
P_{i,\text{off}}=P_{i,\text{off}}^{\text{H}}=1-e^{-\pi\lambda_{\text{H}} p_{i}^{\text{H}} {R_{\text{H}}}^2}
\end{equation}
Problem  (\ref{pro:equvalent_min}) can be written as
\begin{align}\label{case1}
    & \min_{\mathbf{P_{\text{H}}}} ~~-\sum_{i=1}^{N}q_{i}(1-e^{-\pi\lambda_{\text{H}} p_{i}^{\text{H}} {R_{\text{H}}}^2})\\
    & \textup{s.t.}~~\begin{cases}  \sum\limits_{i=1}^{N} {p_{i}^{\text{H}}}\leq M_{\text{H}}\\
  0 \leq p_i^{\text{H}}\leq 1, i\in\{1,\ldots,N\}\\
    \end{cases},
\end{align}
\begin{lemma} (\textbf{Water-filling method}) The optimal caching placement of the helpers is
\begin{equation}
p_{i}^{\text{H}}=\min\left((\beta+\frac{\ln {q_{i}}}{\pi \lambda_{\text{H}} {R_{\text{H}}}^2})^+,1\right)
\end{equation}
for $i=1,2,\ldots,N$, where  $x^+=\max{(x,0)}$ and $\beta$  is effectively solved by the bisection search with $\sum\limits_{i=1}^{N} {p_{i}^{\text{H}}}=M_{\text{H}}$.
\end{lemma}
\begin{proof}
The second derivative of $P_{i,\text{off}}$ is
\begin{align}
\frac{\partial^2{P_{i,\text{off}}}}{\partial{{p_{i}^{\text{H}} }^2}}=-\pi^2\lambda_{\text{H}}^2{R_{H}^{4}}e^{-\pi\lambda_{\text{H}} p_{i}^{\text{H}} {R_{\text{H}}}^2}< 0,
 \end{align}
 thus  $-P_{i,\text{off}}$ is convex in $p_{i}^{\text{H}}$ and the objective function $-\sum\limits_{i=1}^{N}q_{i}P_{i,\text{off}}$ is also convex. Therefore, the caching placement optimization  problem  is  a convex problem.  Consider the following Lagrangian
\begin{align}
{\cal{L}}= -\sum_{i=1}^{N}q_{i}(1-e^{-\pi\lambda_{\text{H}} p_{i}^{\text{H}} {R_{\text{H}}}^2})+\mu(\sum\limits_{i=1}^{N} {p_{i}^{\text{H}}}- M_{\text{H}})
\end{align}
where $\mu$ is the Lagrange multiplier. The KKT condition for the optimality of a caching placement is
 \begin{equation}
\frac{\partial{\cal{L}}}{\partial{p_{i}^{\text{H}}}}=-\pi\lambda_{\text{H}}{R_{\text{H}}}^2q_{i}e^{-\pi\lambda_{\text{H}} p_{i}^{\text{H}} {R_{\text{H}}}^2}+\mu\begin{cases}=0 \quad \text{if } 0<p_i^{\text{H}}<1\\
\geq 0 \quad \text{if } p_i^{\text{H}}=0\\
\leq 0 \quad \text{if } p_i^{\text{H}}=1
\end{cases}.
\end{equation}
Let $\beta=\frac{\ln {(\pi \lambda_{\text{H}}{R_{\text{H}}}^2)}-\ln \mu }{\pi \lambda_{\text{H}} {R_{\text{H}}}^2}$ and $x^+=\max{(x,0)}$, we then have
\begin{equation}
p_{i}^{\text{H}}=\min\left((\beta+\frac{\ln {q_{i}}}{\pi \lambda_{\text{H}} {R_{\text{H}}}^2})^+,1\right),
\end{equation}
where $\beta$ can be solved by  the bisection search method under the cache storage constraint.
\end{proof}
\begin{figure}[t]
\centering
\includegraphics[width=3.2in]{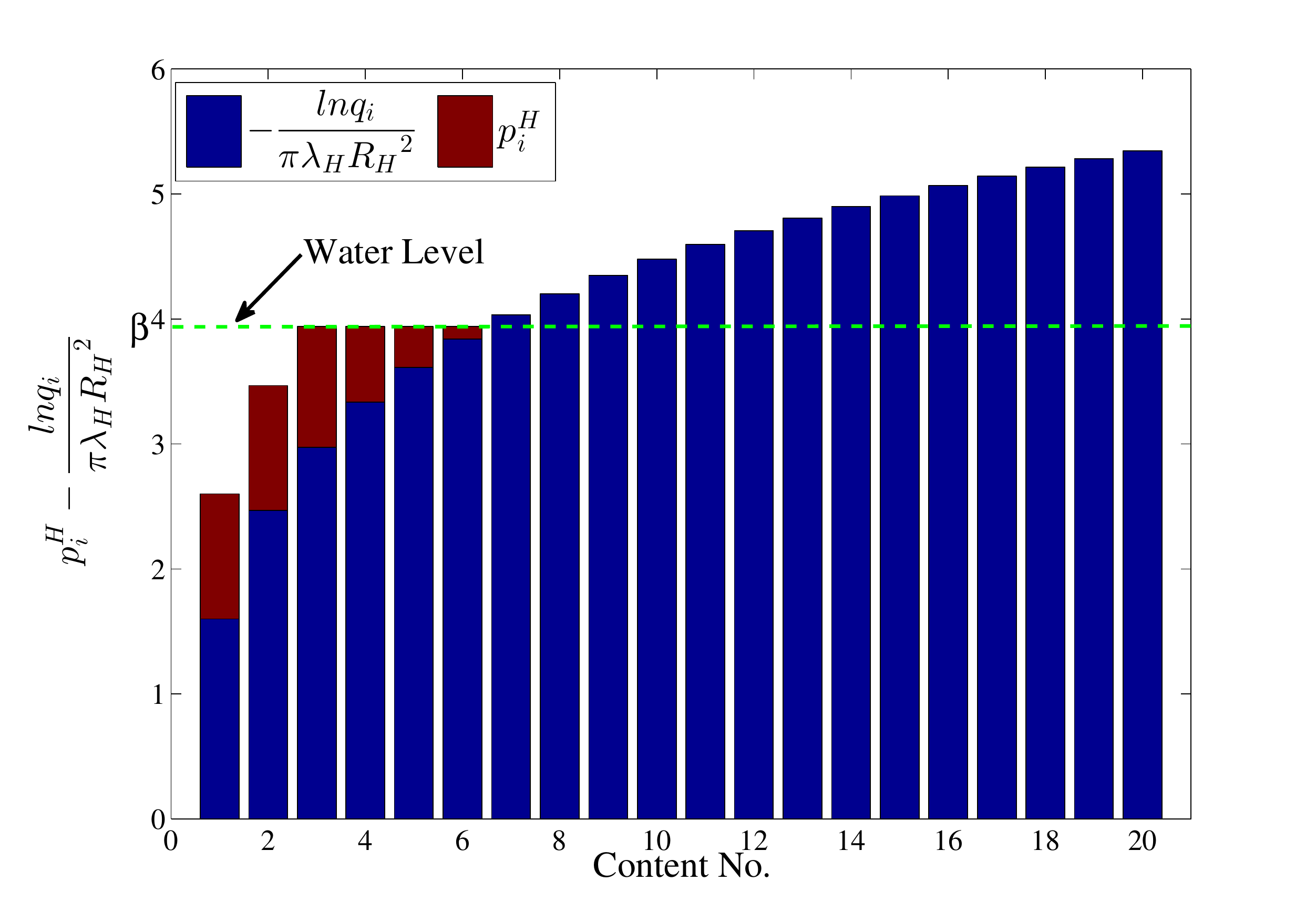}
\caption{Optimal caching placement at helper side under the settings $N=20,
\gamma=1,M_{\text{H}}=4,\lambda_{\text{H}}=\frac{20}{\pi500^2}$.}
\label{fig:waterfilling}
\end{figure}

As illustrated in Fig. \ref{fig:waterfilling}, the water-filling method  allocate more cache probability to contents with larger popularity.  For instance, the contents with larger popularities under the water level, i.e., the first content and the second contents have been cached in all  helpers. While the contents with smaller popularity above the water level, i.e., the $7$-th content to the last content,  have not been cached in any helper.

\begin{remark}
According to  \textbf{Lemma 1},  it is straightforward that the most popular contents are cached in helper storage under relatively low helper density i.e., $p_{1}^{\textbf{H}}=\ldots=p_{M_\textbf{H}}^{\textbf{H}}=1$ and  $p_{M_\textbf{H}+1}^{\textbf{H}}=\ldots=p_{N}^{\textbf{H}}=0$.  While under relatively high density, contents are evenly cached at helper side, i.e, $p_{1}^{\textbf{H}}=\ldots=p_{N}^{\textbf{H}}=\frac{M_{\text{H}}}{N}$.
\end{remark}
\subsection{$\lambda_{\text{H}}M_{\text{H}}=0$: \textbf{user-tier  caching network}}
In this case, no helpers with caching ability participate in offloading the user requests,  the optimization problem is reduced to the optimal caching placement of $p_{i}^{\text{UE}}$. We thus rewritten the function of offloading probability as:
\begin{equation}
P_{i,\text{off}}=1-(1-\alpha{p_{i}^{\text{UE}}})e^{-\pi\alpha \lambda_{\text{UE}} p_{i}^{\text{UE}} {R_{\text{UE}}}^2}.
\end{equation}
Then problem (\ref{pro:equvalent_min}) becomes
\begin{align}\label{case2}
    & \min_{\mathbf{P_{\text{UE}}}} ~~-\sum_{i=1}^{N}q_{i}(1-(1-\alpha{p_{i}^{\text{UE}}})e^{-\pi\alpha \lambda_{\text{UE}} p_{i}^{\text{UE}} {R_{\text{UE}}}^2})\\
    & \textup{s.t.}~~\begin{cases}  \sum\limits_{i=1}^{N} {p_{i}^{\text{UE}}}\leq M_{\text{UE}}\\
  0 \leq p_i^{\text{UE}}\leq 1, i\in\{1,\ldots,N\}\\
    \end{cases},
\end{align}
\begin{proposition}
The above problem is also a convex problem.
\end{proposition}
\begin{proof}
The second derivative of the objective function becomes
\begin{align}
\frac{\partial^2{P_{i,\text{off}}}}{\partial{{p_{i}^{\text{UE}} }^2}}=-[2\alpha{b}+b^2(1-\alpha{p_{i}^{\text{UE}}})]e^{-bp_{i}^{\text{UE}}}< 0,
 \end{align}
 where $b=\pi\alpha \lambda_{\text{UE}}{R_{\text{UE}}}^2$. Therefore,  $-P_{i,\text{off}}$ is convex in $p_{i}^{\text{UE}}$ and the objective function $-\sum\limits_{i=1}^{N}q_{i}P_{i,\text{off}}$ is also convex. Therefore, the caching placement optimization  problem  is convex.
 \end{proof}

As a result, we can adopt a inter-point method to achieve the optimal solution\cite{convex}.
\section{Simulations}\label{sec:result}

\begin{table}
  \centering
  \caption{Default parameter setting}
  \label{table:parameter}
  \begin{tabular}{|c|c|}
    \hline
    Parameters& values  \\
    \hline
    D2D communication range:$R_{\text{UE}}$ & 15(m) \\
    helper transmission range:$R_{\text{H}}$ & 100 (m) \\
    the proportion of cache-enabled users:$\alpha$&0.5 \\
    the density of users:$\lambda_{\text{UE}}$& $5000/{\pi 500^2}$\\
    the density of helpers:$\lambda_{\text{H}}$&$50/{\pi 500^2}$ \\
    the cache capacity of users and helpers&$M_{\text{UE}}=2;M_{\text{H}}=8$ \\
    the size of content library:N& 30\\
   the skewness of the popularity:$\gamma$&1\\
    \hline
  \end{tabular}
\end{table}

\begin{table}
  \centering
  \begin{center}
   \caption{Different caching schemes}
\label{table:baseline}
  \begin{tabular}{|c|c|c|}
    \hline
    Schemes & caching schemes of users & caching schemes of helpers \\
    \hline

{popular cache}& $p_i^{\text{UE}}=1,i\in[1,M_{\text{UE}}] $ &  $p_i^{\text{H}}=1,i\in[1,M_{\text{H}}]$ \\
                            & $p_j^{\text{UE}}=0, j\in[M_{\text{UE}}+1,N]$ &  $p_j^{\text{H}}=0, j\in[M_{\text{H}}+1,N]$ \\
    \hline
     {even cache}& $\mathbf{P}^{\text{UE}}=M_{\text{UE}}/N$ & $\mathbf{P}^{\text{H}}=M_{\text{H}}/N$ \\
    \hline
    {Non-joint}&the solution of Problem (\ref{case2})&the solution of problem (\ref{case1})\\
    \hline
  \end{tabular}
  \end{center}
\end{table}

\begin{figure}[t]
\centering
\includegraphics[width=2.8in]{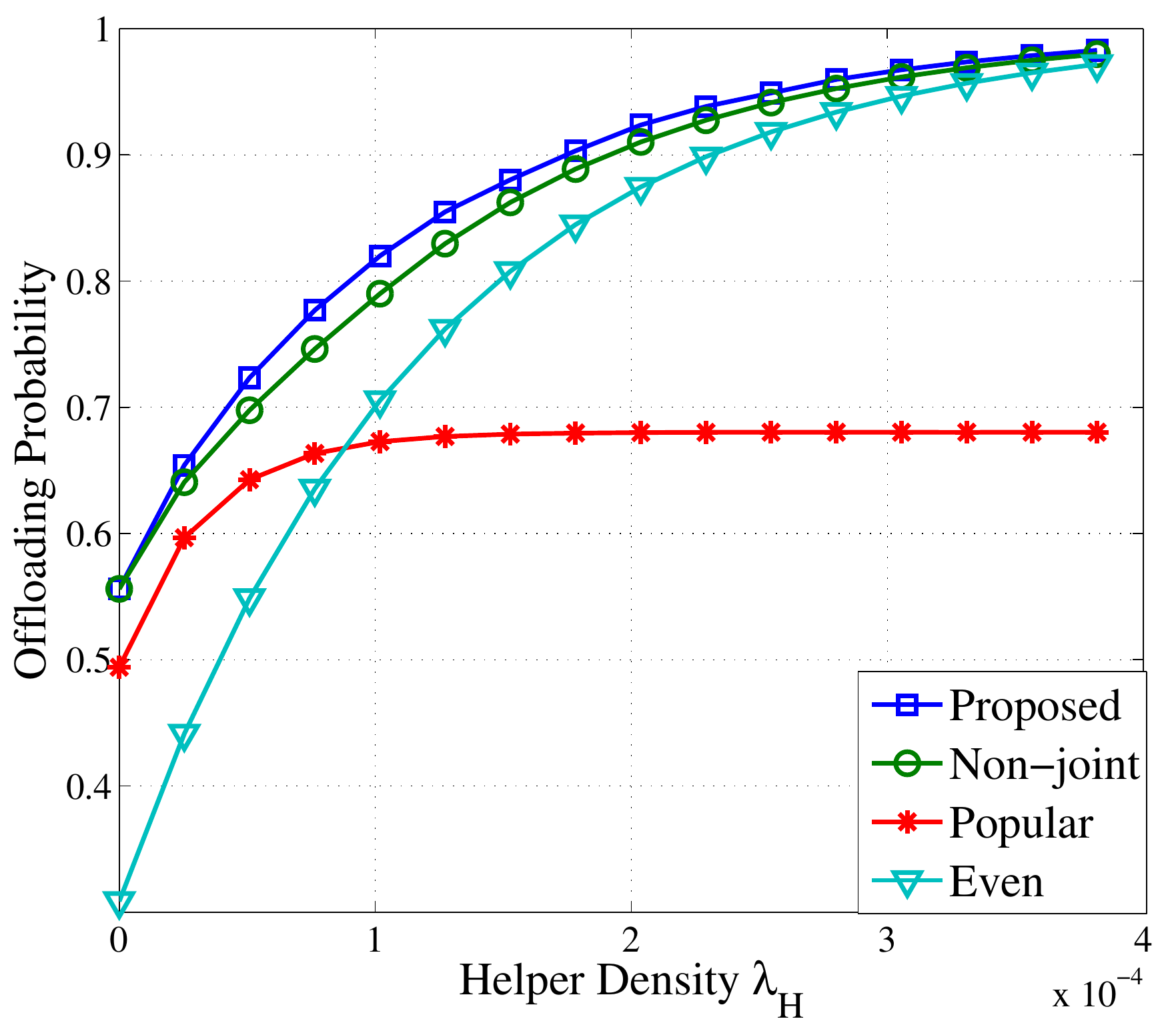}
\caption{The impact of $\lambda_{\text{H}}$ on the offloading probability}
\label{fig:lambda2}
\end{figure}

\begin{figure}[t]
\centering
\includegraphics[width=2.8in]{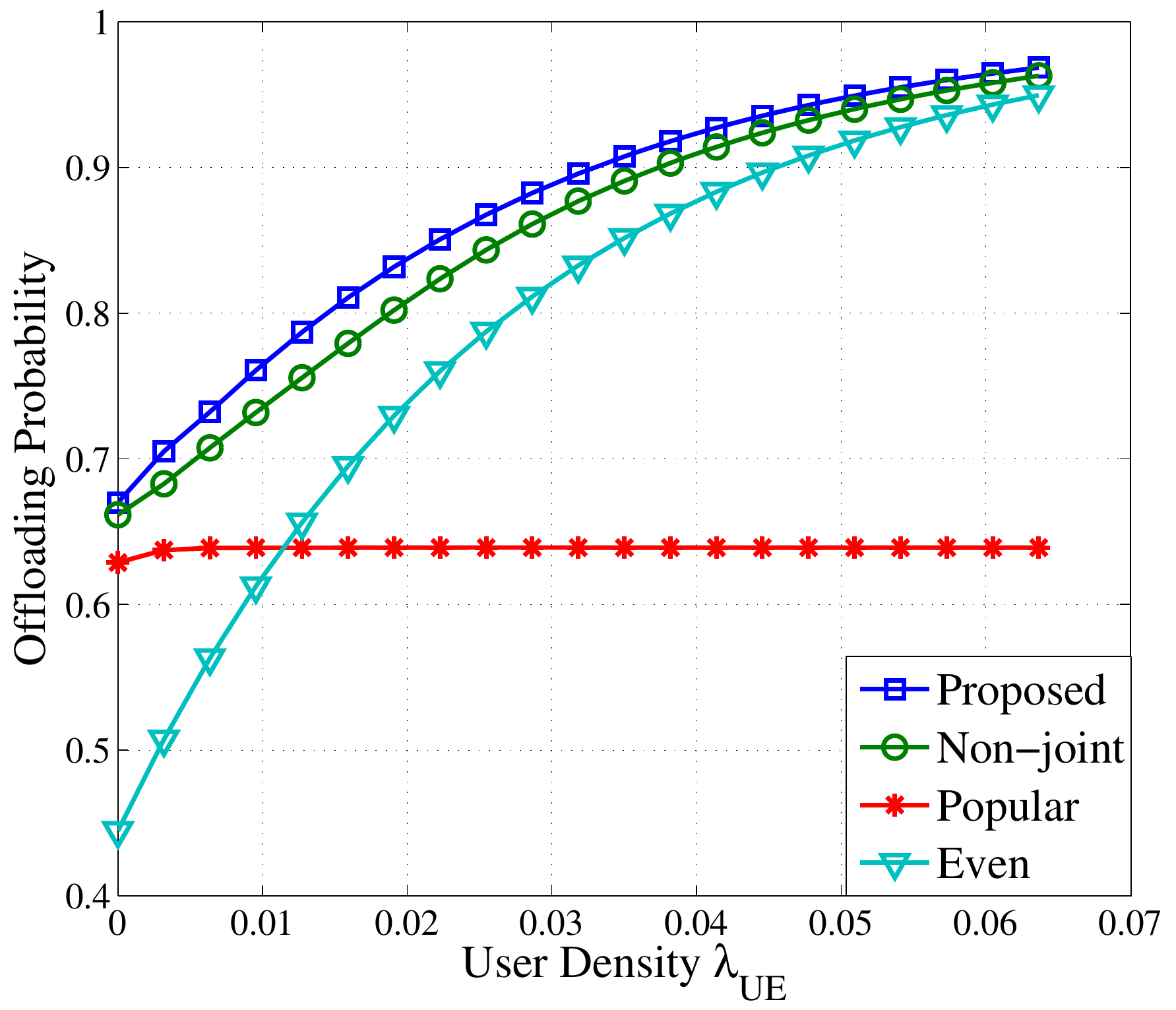}
\caption{The impact of $ \lambda_{\text{UE}}$ on the offloading probability}
\label{fig:lambda1}
\end{figure}

\begin{figure}[t]
\centering
\includegraphics[width=3.1in]{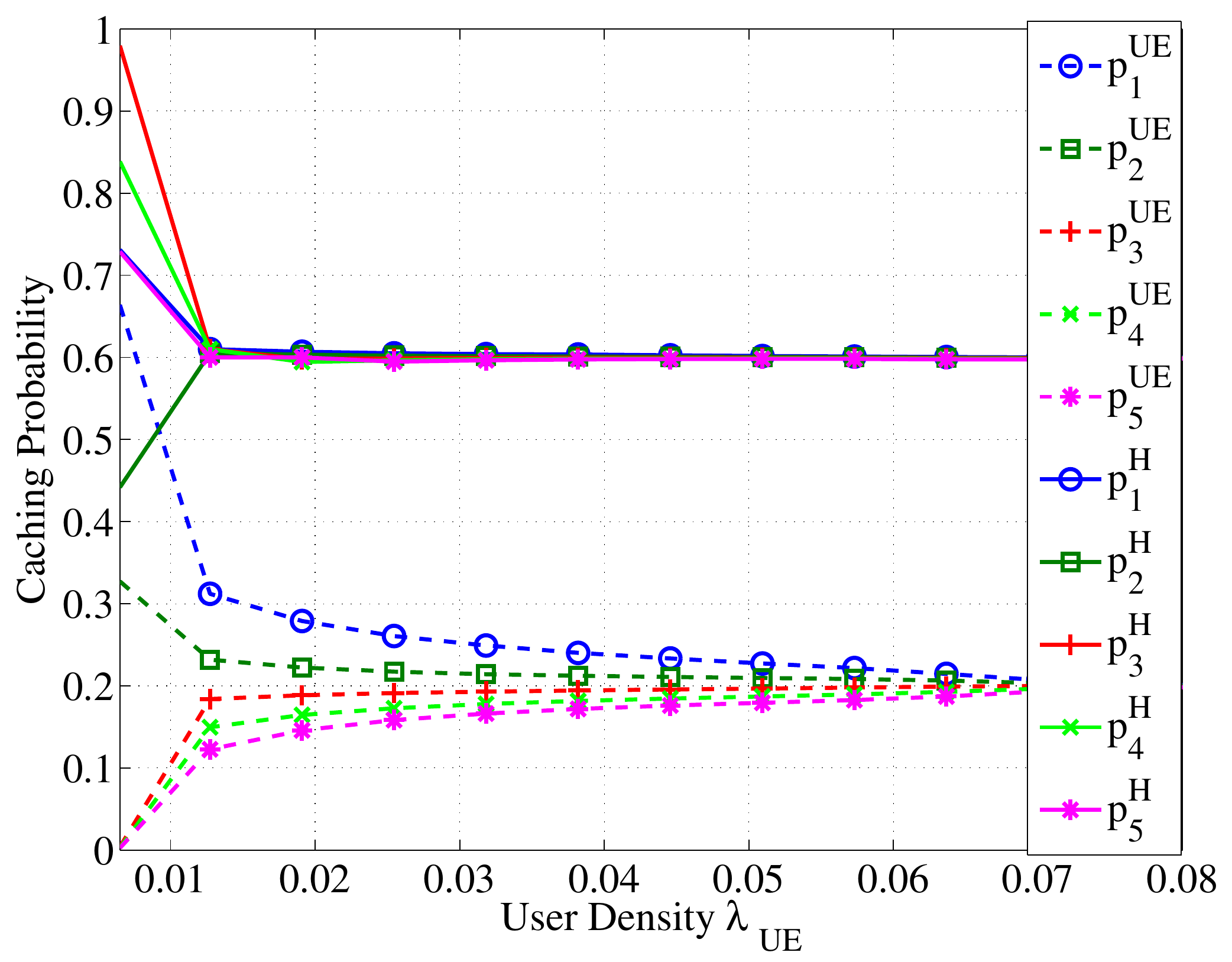}
\caption{the caching placement of the proposed placement}
\label{fig:distribution}
\end{figure}

\begin{figure}[t]
\centering
\includegraphics[width=2.8in]{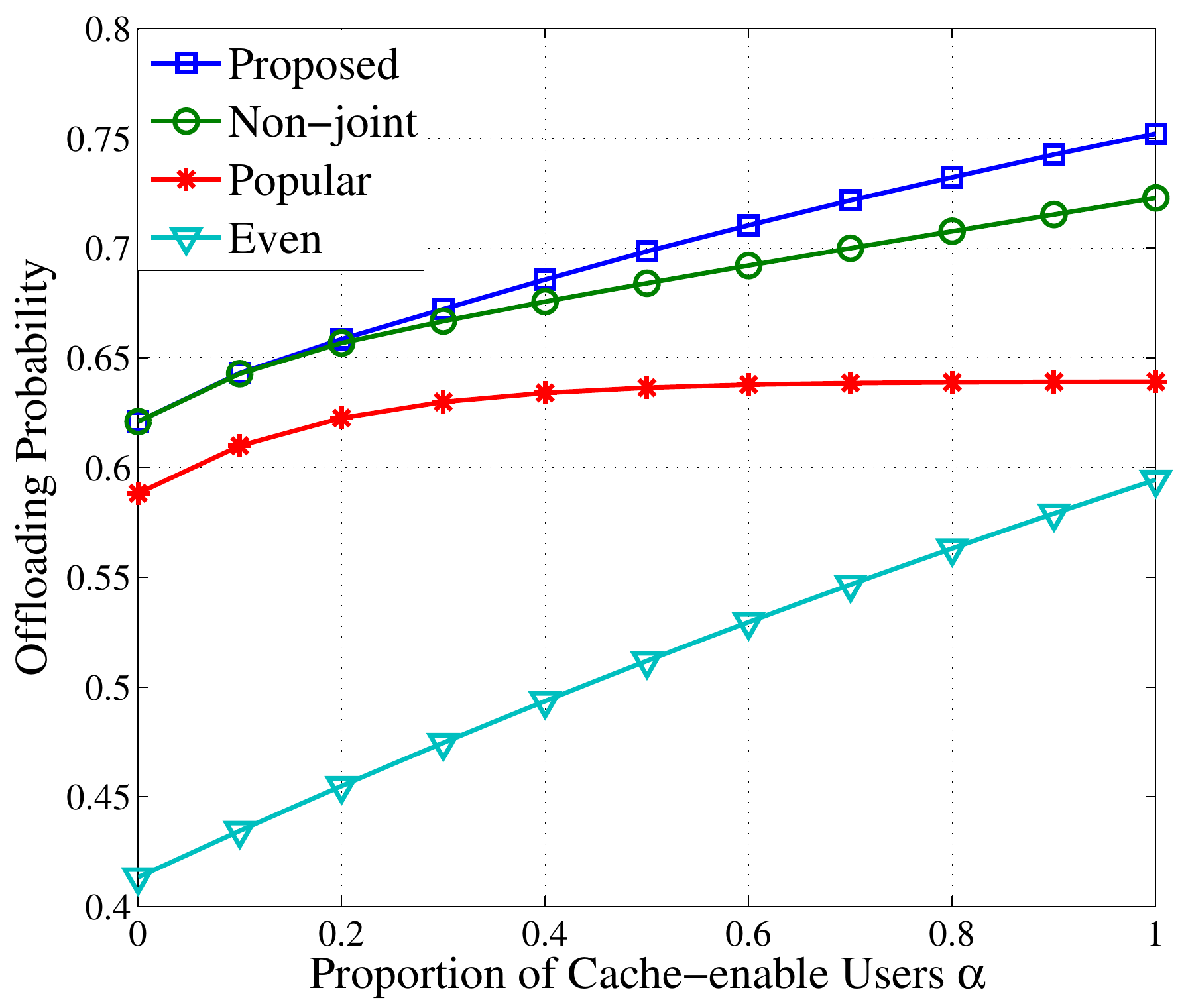}
\caption{The impact of $\alpha$ on the offloading probability}
\label{fig:beta}
\end{figure}

\begin{figure}[t]
\centering
\includegraphics[width=2.8in]{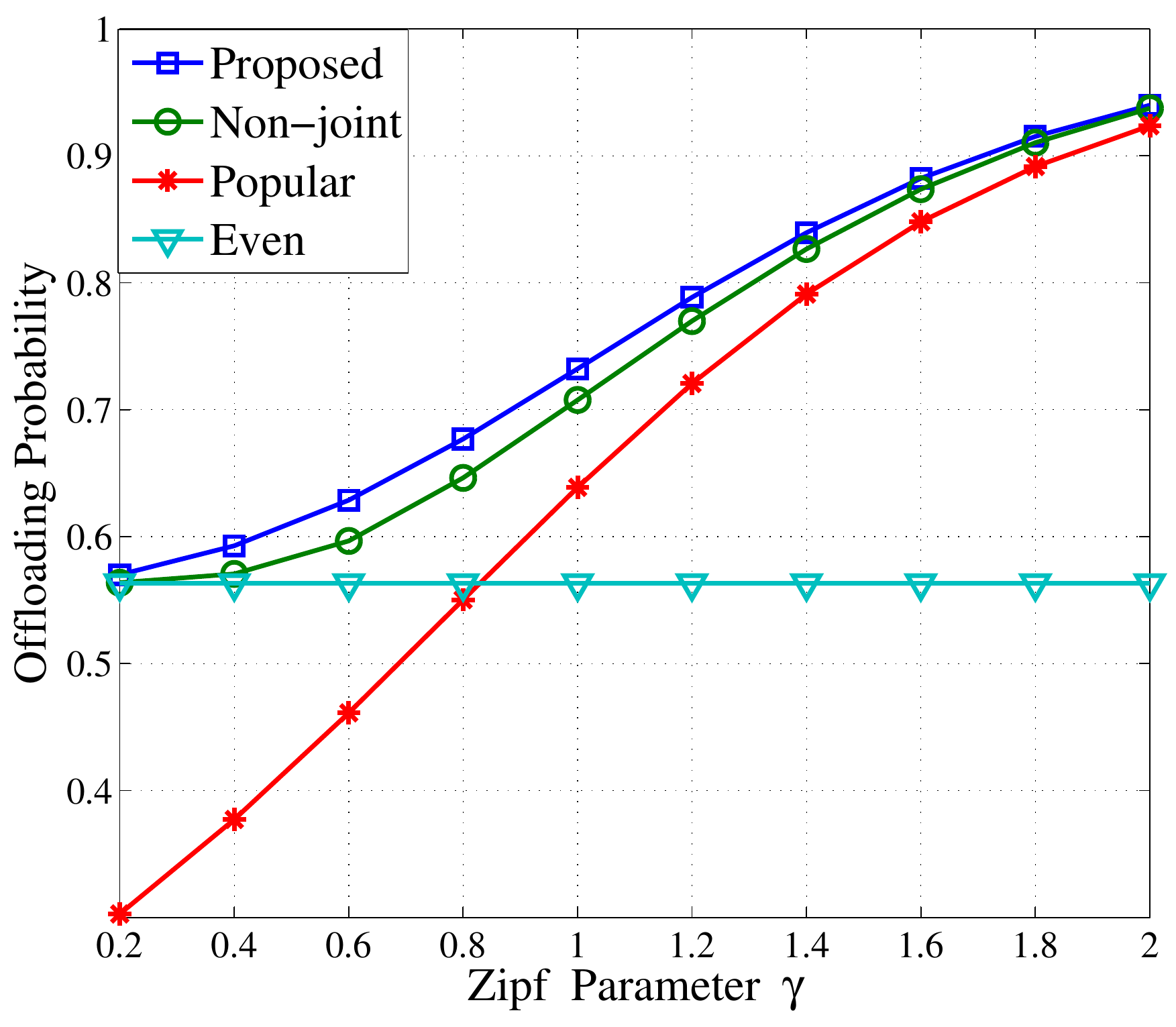}
\caption{The impact of $\gamma$ on the offloading probability}
\label{fig:alpha}
\end{figure}

\begin{figure}[t]
\centering
\includegraphics[width=2.8in]{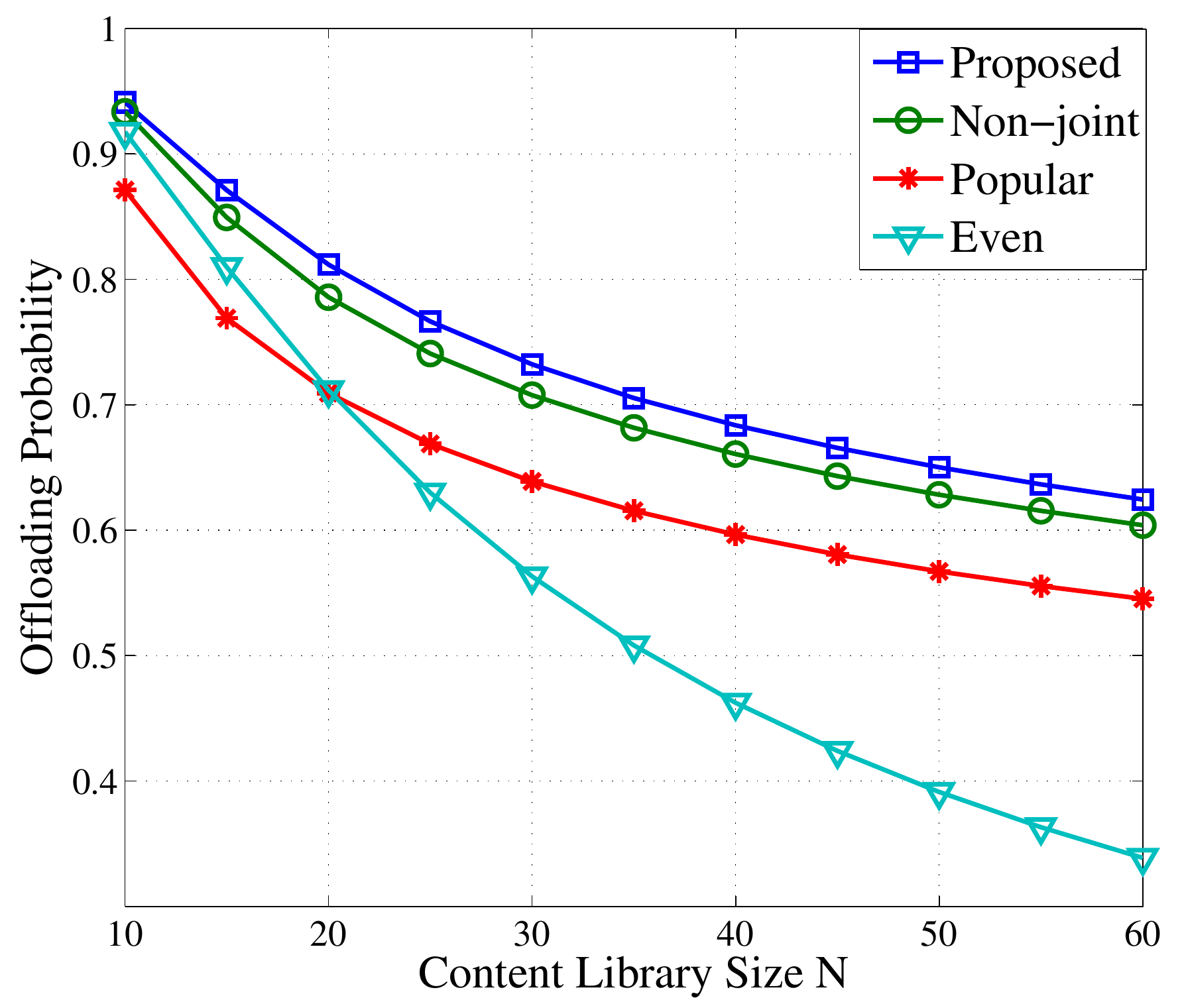}
\caption{The impact of N on the offloading probability}
\label{fig:M}
\end{figure}

In this section, we provide some numerical results to verify our analysis and compare the performance of the proposed caching placement with other baselines. Parameter setting and the three baselines are described in Table \ref{table:parameter} and Table \ref{table:baseline}. In particular, we combine the optimal solutions of the two one-tier caching cases as a baseline and named it non-joint optimal caching placement.

   Fig. \ref{fig:lambda2} shows that the offloading probability increases with helper density $\lambda_{\text{H}}$. The performance of the  proposed caching placement is better than other three baselines no matter how $\lambda_{\text{H}}$  changes. When $\lambda_{\text{H}}=0$, the performance of the proposed caching placement is equal to the non-joint one, because in this situation there are no helpers joining to offload traffic data. With the increasing of $\lambda_{\text{H}}$, the performance of the proposed caching placement becomes better than the non-joint one. While $\lambda_{\text{H}}$ is considerable large, non-joint caching placement approaches to the proposed scheme again. That is because the caching placement of non-joint schemes is also a optimal solution when there is only helper-tier, and the offloading is mostly consisted of helper-tier in this situation. From Fig. \ref{fig:lambda1} we can draw a similar conclusion about $\lambda_{\text{UE}}$.

Furthermore, from Fig. \ref{fig:lambda2} and Fig. \ref{fig:lambda1}, we can see that when both of $\lambda_{\text{H}}$ and $\lambda_{\text{UE}}$ are small, the performance of even cache scheme is the worst one. As $\lambda_{\text{H}}$ or $\lambda_{\text{UE}}$ increases, the performance of even cache scheme becomes better. When $\lambda_{\text{H}}$=$0.8\times10^{-4}$ and $\lambda_{\text{UE}}$=$1.2\times10^{-2}$, it exceeds over the popular cache scheme. That is because while there are few devices participating in the caching network, users and helpers need to cache popular contents to cope with the corresponding high request probability, thus the popular cache scheme performs well; When the resource of the caching network is rich i.e the node density is relatively high, the offloading probabilities for the most popular contents are easily satisfied, and the surplus storage can be used to cache other relatively less popular contents. So the even cache scheme performs better and the offloading probability of popular cache scheme no longer increases. When $\lambda_{\text{UE}}$ is considerable large, the performance of even cache scheme approaches to the optimal caching placement. In Fig. \ref{fig:distribution}, we demonstrate the proposed caching placement which is calculated by DC programming where $N=5,M_{\text{UE}}=1,M_{\text{H}}=3$. As $\lambda_{\text{UE}}$ increases, we can see that the optimal caching placement changes from a popular cache scheme to a even cache scheme which is consistent with our analysis.

Fig. \ref{fig:beta} illustrates the impact of $\alpha$ on offloading probability, where $\alpha$ stands for the proportion of cache-enabled users. When $\alpha=0$, the system reduces to a helper-tier caching system hence, the offloading probability of the proposed caching placement is equal to the non-joint one. While $\alpha$ is larger, there are more cache-enabled users joining in the caching system and therefore the offloading probability will increase. And we can see that while $\alpha \neq0$, the performance of the proposed placement is clearly better than the non-joint one.

 In this paper, $\gamma$ is denoted as the skewness of content popularity. While $\gamma$ is large, the user requests focus on the popular contents and the caching system may have large probabilities to cache the "right" contents. Therefore the offloading probability usually increases with $\gamma$ and we show it in Fig. \ref{fig:alpha}. The performance of popular cache scheme grows rapidly with increasing $\gamma$ while the performance of even cache scheme is not affected by $\gamma$, because it caches every content with a same probability.

Fig. \ref{fig:M} illustrates that the offloading probability decreases with $N$. To expand the size of content library $N$, in a sense, is similar to reduce the cache capacity $M$, thus the offloading probability will experience a decline accordingly. However we can notice that the performance of our proposed caching placement is still well. It demonstrates that when the system is applied into a multi-contents situation, the proposed caching placement can finely adjust the caching proportion of each content by a joint optimization and keep a good performance.

\section{conclusion}

In this paper, the  optimal caching  placement are proposed to  maximize  the  total offloading  probability for the D2D assisted  wireless caching network. Specifically, the caching placement problem for the two-tier caching network is formulated as a DC problem and be solved by the DC programming. In addition, the extreme analysis are provided for helper-tier (or user-tier) caching case in absence of the other tier. The caching placements for both cases are proved to be convex. Moreover, the classical water-filling method is adopted to solve the helper-tier caching case. Simulation results indicate the most popular contents are ought to be cached under relatively low node density, while contents are ought to be cached evenly under relatively high node density. And our proposed caching placement can always make a balance of that.
\bibliographystyle{IEEEtran}
\bibliography{paper}
\end{document}